\documentclass[]{llncs}

\usepackage{amsmath}
\usepackage{amssymb}

\usepackage{tikz}
\usetikzlibrary{backgrounds}
\pgfdeclarelayer{myback}
\pgfsetlayers{background,myback,main} 

\usepackage{hyperref}

\newcommand*\samethanks[1][\value{footnote}]{\footnotemark[#1]}

\title{Quantum Lower Bound for Graph Collision Implies Lower Bound for Triangle Detection} 
\author{Kaspars Balodis\thanks{The research leading to these results has received
funding from the European Union Seventh Framework Programme (FP7/2007-2013)
under grant agreement n%
${{}^\circ}$
600700 (QALGO), ERC Advanced Grant MQC, and Latvian State Research programme NexIT project No.1.} \and J\={a}nis Iraids\samethanks}
\institute{Faculty of Computing, University of Latvia, Rai\c{n}a bulv\={a}ris 19, R\={\i}ga, LV-1586, Latvija}
\date{\today}

\begin{document}

\maketitle

\begin{abstract}
We show that an improvement to the best known quantum lower bound for \textsc{Graph-Collision} problem implies an improvement to the best known lower bound for \textsc{Triangle} problem in the quantum query complexity model. In \textsc{Graph-Collision} we are given free access to a graph $(V,E)$ and access to a function $f:V\rightarrow \{0,1\}$ as a black box. We are asked to determine if there exist $(u,v) \in E$, such that $f(u)=f(v)=1$. In \textsc{Triangle} we have a black box access to an adjacency matrix of a graph and we have to determine if the graph contains a triangle. For both of these problems the known lower bounds are trivial ($\Omega(\sqrt{n})$ and $\Omega(n)$, respectively) and there is no known matching upper bound.
\end{abstract}

\section{Introduction}

By $Q(f)$ we denote the bounded-error quantum query complexity of a function $f$.
We consider the quantum query complexity for some graph problems.

\begin{definition}
In \textsc{Triangle} problem it is asked whether an $n$-vertex graph $G=(V,E)$ contains a triangle, i.e. a complete subgraph on three vertices.
The adjacency matrix of the graph is given in a black box which can be queried by asking if $(x,y) \in E$.
\end{definition}

Recently there have been several improvements in the algorithms for the \textsc{Triangle} problem in the quantum black box model. The problem was first considered by Buhrman et al. in 2005 \cite{Betal05} who gave an $O(n+\sqrt{nm})$ algorithm where $n$ is the number of vertices and $m$ -- the number of edges. Later in 2007 Magniez et al. gave an $\tilde{O}(n^{13/10})$ algorithm based on quantum walks. Introducing a novel concept -- learning graphs, and using a new technique in 2012 Belovs \cite{Bel12} was able to reduce the complexity to $O(n^{35/27})$. In 2013 Lee et al. \cite{LMS13} using a more refined learning graph approach reduced the complexity to $\tilde{O}(n^{9/7})$. Currently the best known algorithm is by Le Gall who exhibited a quantum algorithm which solves the \textsc{Triangle} problem with query complexity $\tilde{O}(n^{5/4})$ \cite{LeGall14}.
Classically the query complexity of \textsc{Triangle} is $\Theta(n^2)$; however, it is an open question whether \textsc{Triangle} can be computed in time better than $O(n^\omega)$ where $\omega$ is the matrix multiplication constant.

\begin{definition}
In \textsc{Graph-Collision$_G$} problem a known $n$-vertex undirected graph $G=(V,E)$ is given and a coloring function $f:V \rightarrow \{0,1\}$ whose values can be obtained by querying the black box for the value of $f(x)$ of a given $x \in V$.
We say that a vertex $x \in V$ is marked iff $f(x)=1$.
The value of the \textsc{Graph-Collision$_G$} instance is $1$ iff there exists an edge whose both vertices are marked, i.e. $\exists (x,y) \in E \  f(x)=f(y)=1$.
\end{definition}

By $Q(\textsc{Graph-Collision})$ we mean the complexity of solving  \textsc{Graph-Collision$_G$} for the hardest $n$-vertex graph $G$.

There has been an increased interest in the quantum query complexity of the \textsc{Graph-Collision} problem, mainly because
algorithms for solving \textsc{Graph-Collision} are used as a subroutine in algorithms for the \textsc{Triangle} problem \cite{MSS07}
and Boolean matrix multiplication \cite{JKM12}.

The best known quantum algorithm for \textsc{Graph-Collision} for an arbitrary $n$-vertex graph has complexity $O(n^{2/3})$ \cite{MSS07}.
However, for some graph classes there are algorithms with complexity $O(\sqrt{n})$ \cite{Amb13,Bel12a,GI12,JKM12}. 
It is an open question whether for every $n$-vertex graph $G$ \textsc{Graph-Collision$_G$} can be solved with $O(\sqrt{n})$ queries.

Contrary to the improvements in the algorithms for these two problems, the best known lower bounds for $Q(\textsc{Graph-Collision})$ and $Q(\textsc{Triangle})$ are still the trivial $\Omega (\sqrt{n})$ and $\Omega (n)$ respectively, which follow from the reduction to search problem.
Nonetheless these lower bounds seem hard to improve with the current techniques.

As mentioned before, algorithms for \textsc{Graph-Collision} have been used as a subroutine for constructing algorithms for the \textsc{Triangle} problem, therefore an improved algorithm for \textsc{Graph-Collision} would result in an improved algorithm for \textsc{Triangle}.
In this paper we show a reduction in the opposite direction---that an improvement in the lower bound on $Q(\textsc{Graph-Collision})$ would imply an improvement in the lower bound on $Q(\textsc{Triangle})$.

\section{Result}

\begin{theorem}
If there is a graph $G = (V, E)$ with $n$ vertices such that \textsc{Graph-Collision$_G$} has quantum query complexity $t$ then \textsc{Triangle} problem has quantum query complexity at least $\Omega(t\sqrt{n})$.
\end{theorem}
\begin{proof}
We show how to transform the graph $G$ into a graph $G'$ with $3n$ vertices so that it is hard to decide if $G'$ contains a triangle.
More precisely, we construct the graph $G'$ in such a way that solving the \textsc{Triangle} problem on $G'$ is equivalent to solving \textit{OR} function from the results of $n$ independent instances of \textsc{Graph-Collision$_G$}.

First, we want to get rid of any triangles in $G$, therefore we transform $G$ into an equivalent bipartite graph $G_2=(V_2, E_2)$ with $2n$ vertices by setting $V_2 = \{ v_1, v_2 \mid v \in V \}$ and $E_2 = \{ (x_1,y_2) \mid (x,y) \in E \}$.
The graph $G_2$ is equivalent to $G$ in the following sense---if we mark the vertices $v_1$ and $v_2$ in $G_2$ for every marked vertex $v$ in $G$, then $G_2$ has a collision iff $G$ has a collision.
However, the graph $G_2$ does not contain any triangle (since it is bipartite).

Next, we add $n$ isolated vertices $z_1, \dots, z_n$ to $G_2$ thereby obtaining a graph $G'$.
Let $f_1,\dots,f_n:V \rightarrow \{0,1\}$ be the colorings from $n$ independent \textsc{Graph-Collision$_{G}$} instances.
We add the edges $(z_i, v_1)$ and $(z_i, v_2)$ to $G'$ iff $v \in V$ is marked by the respective coloring, i.e., iff $f_i(v)=1$.

 See Fig. \ref{graphs} for an example.

\begin{figure}[!htb]
\centering
\begin{tikzpicture}[]
\tikzstyle{vert}=[circle,draw=black,fill=white,inner sep=1.5pt, minimum size=20]

\node[vert] (a1) at (90  :1)   {a};
\node[vert] (a2) at (18 :1)    {b};
\node[vert] (a3) at (-54:1)    {c};
\node[vert] (a4) at (-126:1)   {d};
\node[vert] (a5) at (-198:1)   {e};

\draw (a1) -- (a2);
\draw (a2) -- (a3);
\draw (a3) -- (a4);
\draw (a3) -- (a5);
\draw (a4) -- (a5);

\begin{scope}[xshift=6cm]

\node[vert] (a1) at (-0.75, 5)   {$a_1$};
\node[vert] (a2) at (-0.75, 4)   {$b_1$};
\node[vert] (a3) at (-0.75, 3)   {$c_1$};
\node[vert] (a4) at (-0.75, 2)   {$d_1$};
\node[vert] (a5) at (-0.75, 1)   {$e_1$};

\node[vert] (b1) at (0.75, 5)   {$a_2$};
\node[vert] (b2) at (0.75, 4)   {$b_2$};
\node[vert] (b3) at (0.75, 3)   {$c_2$};
\node[vert] (b4) at (0.75, 2)   {$d_2$};
\node[vert] (b5) at (0.75, 1)   {$e_2$};

\draw (a1) -- (b2);
\draw (a2) -- (b3);
\draw (a3) -- (b4);
\draw (a3) -- (b5);
\draw (a4) -- (b5);

\draw (b1) -- (a2);
\draw (b2) -- (a3);
\draw (b3) -- (a4);
\draw (b3) -- (a5);
\draw (b4) -- (a5);

\node[vert] (z1) at (-2, -1)   {$z_1$};
\node[vert] (z2) at (-1, -1)   {$z_2$};
\node[vert] (z3) at ( 0, -1)   {$z_3$};
\node[vert] (z4) at ( 1, -1)   {$z_4$};
\node[vert] (z5) at ( 2, -1)   {$z_5$};

\begin{pgfonlayer}{myback}
\draw (z1) -- (a4);
\draw (z1) -- (a5);
\draw (z1) -- (b4);
\draw (z1) -- (b5);
\end{pgfonlayer}

\draw (z2)+( 65:0.65) -- (z2);
\draw (z2)+( 80:0.65) -- (z2);

\draw (z3)+( 70:0.75) -- (z3);
\draw (z3)+( 75:0.75) -- (z3);
\draw (z3)+( 80:0.75) -- (z3);
\draw (z3)+(110:0.75) -- (z3);
\draw (z3)+(115:0.75) -- (z3);
\draw (z3)+(120:0.75) -- (z3);

\draw (z4)+(100:0.75) -- (z4);
\draw (z4)+(104:0.75) -- (z4);
\draw (z4)+(120:0.75) -- (z4);
\draw (z4)+(124:0.75) -- (z4);

\draw (z5)+(110:0.75) -- (z5);
\draw (z5)+(115:0.75) -- (z5);
\draw (z5)+(130:0.75) -- (z5);
\draw (z5)+(140:0.75) -- (z5);

\node[] (d) at (0.8, 0.10)   {$\dots$};

\end{scope}

\end{tikzpicture}
\caption{Graph $G$ and the resulting graph $G'$}
\label{graphs}
\end{figure}
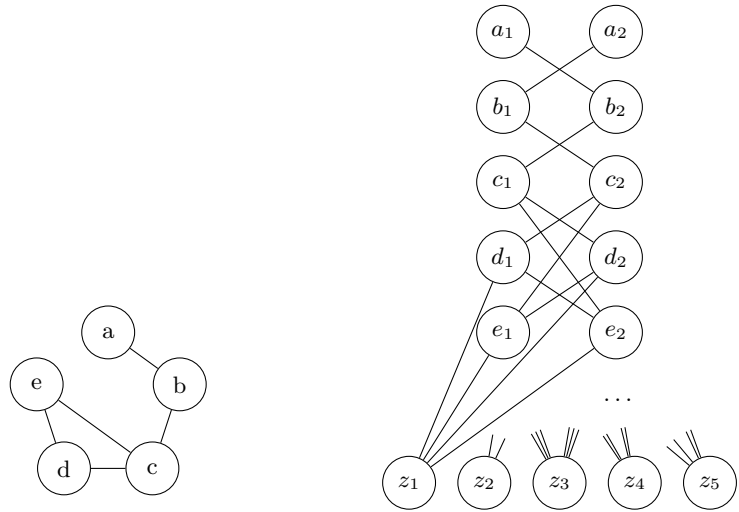

The only possible triangles in the graph $G'$ can be of the form $\{z_i, v_1, w_2\}$ for some $i \in \{1, \dots, n\}$ and $v, w \in V$.
Moreover, there is a triangle $\{z_i, v_1, w_2\}$ iff $f_i$ is such coloring that $G$ has a collision $(v, w)$, i.e., iff $f_i(v)=f_i(w)=1$.
Therefore detecting a triangle in $G'$ is essentially calculating $OR$ function from the results of $n$ instances of \textsc{Graph-Collision$_G$}.

We now use the fact that $OR$ function requires $\Omega(\sqrt{n})$ queries, the assumption that \textsc{Graph-Collision$_G$} requires $t$ queries and the Theorem 1.5. from \cite{Reichardt11}:
\begin{theorem}
Let $f : \{ 0, 1\}^n \rightarrow \{0, 1\}$ and $g : \{ 0, 1\}^m \rightarrow \{0, 1\}$. Then
\[
Q(f \bullet g) = \Theta(Q(f)Q(g)),
\]
where $(f \bullet g) (x)= f( g(x_1, \dots, x_m), \dots, g(x_{(n-1)m+1}, \dots, x_{nm}) )$.
\end{theorem}

Setting $f = OR$ and $g = \textsc{Graph-Collision$_G$}$ gives the desired bound.
\end{proof}

As the next corollary shows, a better lower bound on \textsc{Graph-Collision} implies a better lower bound on the \textsc{Triangle} problem.

\begin{corollary}
If $Q_2(\textsc{Graph-Collision}) = \omega(\sqrt{n})$ then $Q_2(\textsc{Triangle}) = \omega(n)$.
\end{corollary}

\bibliographystyle{abbrv}
\bibliography{gc_tr}

\end{document}